\documentclass[final,5p,times,twocolumn,authoryear]{elsarticle}

\usepackage{amssymb}
\usepackage{amsmath}

\journal{Systems \& Control Letters}

\usepackage{amsmath,amsfonts, amsthm}
\usepackage{algorithmic}
\usepackage{array}
\usepackage[caption=false,font=normalsize,labelfont=sf,textfont=sf]{subfig}
\usepackage{textcomp}
\usepackage{stfloats}
\usepackage{url}
\usepackage{verbatim}
\usepackage{graphicx}
\usepackage{tikz}
\usetikzlibrary{calc, arrows.meta}
\usepackage[hidelinks]{hyperref}
\usepackage{glossaries}
\usepackage{enumitem}

\newtheorem{theorem}{Theorem}
\newtheorem{definition}{Definition}

\newtheorem{corollary}{Corollary}
\newtheorem{proposition}{Proposition}

\newtheorem{remark}{Remark}

\newcommand{\R}{\mathbb{R}}
\newcommand{\C}{\mathbb{C}}

\renewcommand{\Re}[1]{\mathfrak{Re}(#1)}
\newcommand{\trans}{*}

\newcommand{\w}{\omega}

\newcommand{\sym}[1]{\textnormal{H}\{#1\}}
\renewcommand{\skew}[1]{\textnormal{S}\{#1\}}
\newcommand{\spec}[1]{\textnormal{spec}\{#1\}}
\renewcommand{\ker}[1]{\textnormal{ker}\{#1\}}
\newcommand{\image}[1]{\textnormal{im}\{#1\}}
\newcommand{\lmin}[1]{\lambda_{\min}\{#1\}}
\renewcommand{\det}[1]{\mathrm{det}\{#1\}}
\newcommand{\adj}[1]{\mathrm{adj}\{#1\}}

\newcommand{\ifp}{\nu}
\newcommand{\ofp}{\rho}
\newcommand{\Ifp}{\nu'}
\newcommand{\Ofp}{\rho'}
\newacronym{ifp}{IFP}{input feedforward passivity}
\newacronym{ofp}{OFP}{output feedback passivity}
\newacronym{if-ofp}{IF-OFP}{input feedforward-output feedback passive}
\newacronym{fifp}{$\w$-IFP}{frequency-dependent Input Feedforward Passivity}
\newacronym{fofp}{$\w$-OFP}{frequency-dependent Output Feedback Passivity}

\begin{document}

\begin{frontmatter}

\title{A Feedback Stability Theorem for Frequency-dependent Compensation of Excess and Lack of Passivity}

\author[ABB]{\corref{cor1}Pol Jane-Soneira} 
\cortext[cor1]{Corresponding author}
\author[ABB]{Gösta Stomberg}
\author[ABBCH]{Ognjen Stanojev}
\author[ABBCH]{Orcun Karaca}
\author[ABBSE]{Lennart Harnefors}

\affiliation[ABB]{organization={ABB Corporate Research},
	addressline={Kallstadter Str. 1, 68309 Mannheim}, 
	country={Germany}}
\affiliation[ABBCH]{organization={ABB Corporate Research},
	addressline={Segelhofstr. 1K, 5405 Baden}, 
	country={Switzerland}}
\affiliation[ABBSE]{organization={ABB Corporate Research},
	addressline={72226 Västerås}, 
	country={Sweden}}

\begin{abstract}
This article studies the stability of feedback interconnections of linear time-invariant systems based on frequency-dependent passivity indices.
Using these frequency-dependent passivity indices, we show that the feedback interconnection of two systems can be certified to be stable even if both systems have a lack of passivity in terms of their scalar passivity indices.
The main contribution of this paper is a new stability theorem based on frequency-dependent passivity indices.
Moreover, we discuss the connection of the proposed feedback stability theorem to prior results based on scalar passivity indices.
A numerical case study showcases the advantages of frequency-dependent passivity indices over scalar indices for feedback interconnections of linear systems.
\end{abstract}


\end{frontmatter}
\let\thefootnote\relax\footnotetext{\emph{Email addresses}: \{pol.jane-soneira, goesta.stomberg\}@de.abb.com, \{ognjen.stanojev, orcun.karaca\}@ch.abb.com, lennart.harnefors@se.abb.com}


\section{Introduction} \label{sec:intro}

Based on the notion of dissipative dynamical systems first proposed by \cite{willems1972dissipative}, passivity has been established as a powerful concept for the analysis and design of networks of interconnected systems~\citep{arcak_networks_2016}.
For instance, passivity allows to infer stability of interconnected systems based on conditions that are formulated for the individual subsystems.
Crucially, these conditions can be verified on a subsystem level, which simplifies the stability assessment of large-scale interconnected systems, since a full system-level analysis is thus not required. 
Standard passivity theorems assume that each subsystem is passive to guarantee stability of the interconnected system~\citep{sepulchre_constructive_1997, vanderschaft_l-gain_2017}, which limits their applicability, since many systems are not passive in practical applications.

This shortcoming can be addressed through the concepts of \gls{ifp} and \gls{ofp}.
These concepts were first introduced by \cite{hill_stability_1976} to further characterize the stability properties of dissipative systems. 
The associated scalar \gls{ifp} and \gls{ofp} indices quantify the excess or lack of passivity of a system and allow to guarantee the stability of the interconnected system if the excess of passivity of one system compensates the lack of passivity of another system, e.g., \cite[Theorem 2.30]{sepulchre_constructive_1997}. 
Scalar passivity indices have been successfully applied in various domains, including power systems~\citep{yang_distributed_2020}, power electronics~\citep{chen_limitations_2024}, event-triggered control~\citep{zhu2016passivity}, or optimization~\citep{li2020input}. 
Nowadays, there exists a comprehensive theory about stability for interconnected systems exploiting scalar passivity indices, see, e.g.~\citep{sepulchre_constructive_1997,vanderschaft_l-gain_2017,lozano2013dissipative,moylan_stability_1978}.

To reduce the conservatism of scalar passivity indices, frequency-dependent passivity indices for linear systems have been introduced in areas such as robust control \citep{wen_robustness_1988} and power electronics \citep{chen_limitations_2024,harnefors2015passivity}. These indices capture excess and lack of passivity across frequencies and therefore provide a finer characterization than scalar indices. In particular, scalar passivity indices can be recovered as the minimum (i.e., worst-case value) of the frequency-dependent indices over all frequencies~\citep{xia_sector_2020}. Consequently, a scalar index may indicate a passivity lack even when a system shows passivity excess over most frequencies, with lack occurring only in a narrow band.

Despite the advantageous properties of frequency-dependent passivity indices, a comprehensive theory including feedback stability theorems exploiting the frequency dependency has, to the best of the authors' knowledge, not been developed thus far. 
In power electronics applications, frequency-dependent passivity indices have been found to be useful for the stability analysis of grid-connected inverters, see, e.g.~\citep{harnefors2015passivity, chen_limitations_2024}. However, only specific systems are considered, and no general conditions on frequency-dependent passivity indices for the stability of a feedback interconnection are derived. Furthermore, when one of the systems has a lack of \gls{ofp}, the methods therein are not applicable. The stability analysis with frequency-dependent passivity indices is extended by~\cite{khong_feedback_2025} to more general frequency-dependent QSR dissipativity. This is used to analyze the feedback stability of mixed (i.e. positive real, bounded real and negative imaginary) systems. However, despite being general in the type of passivity, it suffers from the same drawback as the previous work. If one restricts QSR dissipativity to \gls{ifp} and \gls{ofp} indices, the method is not applicable whenever a system has a lack of \gls{ofp} index at some frequency.

\emph{Contributions:}
This paper develops a general feedback stability theorem based on frequency-dependent passivity indices.
In particular, our result can guarantee the stability of feedback interconnections even if the subsystems exhibit a lack of passivity, as long as the lack of passivity of the different subsystems occurs at different frequencies.
In contrast, prior results based on scalar passivity indices as in \citep{vanderschaft_l-gain_2017,sepulchre_constructive_1997,lozano2013dissipative}, capture the lack and excess of passivity across all frequencies and thus they make stronger assumptions on the subsystems to guarantee stability of the feedback interconnection.

\section{Preliminaries and Problem Statement}

\emph{Notation:} The set of (nonnegative) real numbers is denoted by $\R$ ($\R_{\geq 0}$) and the set of complex numbers is denoted by $\C$.
The transpose of matrix $A \in \C^{n\times n}$ is denoted by $A^\top$ and the complex-conjugate transpose of $A$ is denoted by $A^\trans$. The adjugate of $A$ is denoted by $\adj{A}$ and the determinant of $A$ is denoted by $\det{A}$.
We denote the hermitian part of a matrix $A \in \C^{n\times n}$ by $\sym{A} \doteq (A + A^\trans)/2$ and the skew-hermitian part of $A$ as $\skew{A} \doteq (A - A^\trans)/2$.
Recall that, for any $A \in \mathbb{C}^{n \times n}$, $A = \sym{A} + \skew{A}$.
Moreover, recall that $x^\trans \sym{A} x$ is real valued and $x^\trans \skew{A} x$ is purely imaginary for all $x \in \mathbb{C}^n$. 
The smallest eigenvalue of a hermitian matrix $A \in \C^{n\times n}$ is denoted by $\lambda_{\min}(A) \in \mathbb{R}$.
The spectrum, kernel, and image of $A$ are denoted by $\spec{A}$, $\ker{A}$, and $\image{A}$, respectively.
The Laplace transform variable is denoted by $s = \sigma + j\w \in \C$.
The Laplace transform of a time-domain signal $x(t)$ is denoted as $x(s)$ and the undamped Laplace or Fourier transform $x(j\w)$ is obtained by setting $s=j\omega$. $\Re{s}$ denotes the real part of $s$.

We consider Linear Time-Invariant (LTI) systems
\begin{subequations} \label{eq:ss}
	\begin{align}
		\dot{x}(t) &= A x(t) + B u(t), \qquad x(0) = 0, \\
		y(t) &= C x(t) + D u(t),
	\end{align}
\end{subequations}
where $x \in \R^n$ is the state, $u \in \R^m$ is the input, $y \in \R^m$ is the output, and the real matrices $A$, $B$, $C$, and $D$ are of appropriate dimension.
The transfer function of system~\eqref{eq:ss} is given by
\begin{equation}\label{eq:tf}
	H(s) = C(sI - A)^{-1}B + D \in \mathbb{C}^{m \times m},
\end{equation} where $s\in \mathbb{C}$ is the complex variable.
Note that transfer functions defined from state space realizations via~\eqref{eq:tf} are rational, i.e., all components of the matrix $H(s)$ are rational functions in~$s$.

\begin{definition}[Dissipativity~\citep{hill_stability_1976}] \label{def:dissipativity}
	System~\eqref{eq:ss} is said to be dissipative, if for all admissible inputs $u(\cdot)$ and all $t \geq 0$, there exists a supply rate $w: \mathbb{R}^m \times \mathbb{R}^m \rightarrow \mathbb{R}$ such that
	\[
	\int_0^t w(u(\tau),y(\tau)) \mathrm{d}\tau \geq 0. 
	\] 
	\hfill $\square$
	
\end{definition}

We next define further passivity notions as in~\cite[Chapter 2]{sepulchre_constructive_1997}, \cite{lozano2013dissipative,vanderschaft_l-gain_2017}.

\begin{definition}[Passivity and variants] \label{def:passivity}
    Consider the supply rate $w_{\textnormal{P}}(u,y) \doteq u^\top y - \Ifp y^\top y - \Ofp u^\top u$ with constants $\Ifp, \Ofp \in \R$. A system is called
    \begin{itemize}
        \item[(i)] passive if it is dissipative with respect to the supply rate $w_{\textnormal{P}}(u,y)$ with $\Ifp \geq 0$ and $\Ofp \geq 0$.
        \item[(ii)] \gls{ifp} if it is dissipative with respect to the supply rate $w_{\textnormal{P}}(u,y)$ with $\Ofp = 0$.
        \item[(iii)] \gls{ofp} if it is dissipative with respect to the supply rate $w_{\textnormal{P}}(u,y)$ with $\Ifp = 0$.
        \item[(iv)] \gls{if-ofp} if it is dissipative with respect to the supply rate $w_{\textnormal{P}}(u,y)$ with $\Ifp \in \R$ and $\Ofp \in \R$. 
    \end{itemize}
\end{definition}

\begin{remark}[Alternative definitions of passivity]
	In addition to a supply rate, the original notion of dissipativity introduced by \cite{willems1972dissipative} requires the existence of a \textit{storage function} which depends on the system state.
	The definition of dissipativity used here is due to~\cite{hill_stability_1976} and is more general, since it does not require the existence of a storage function.
	For passivity, even more general definitions can be stated for input-output maps without defining a state space \cite[Chapter 2]{vanderschaft_l-gain_2017}, \cite[Chapter 2]{lozano2013dissipative}.
	For state-space systems, the definitions using input-output maps or storage functions and supply rates are equivalent~\citep[Proposition 3.1.5]{vanderschaft_l-gain_2017}. An overview of the relation between the diferent definitions of passivity can be found in~\cite[Chapter 4]{lozano2013dissipative}.
	\hfill $\square$
\end{remark}

\begin{remark}[Scalar passivity indices]
	We refer to the constants $\Ofp$ and $\Ifp$ as scalar passivity indices.
	In particular, we note that $\Ofp$ and $\Ifp$ do not depend on the frequency $\w$. \hfill $\square$
\end{remark}

\begin{definition}[Properness]
	A system is said to be
	\begin{itemize}
		\item[(i)] proper, if $H(s) \rightarrow D$ as $|s| \rightarrow \infty$ for some $D \in \R^{m \times m}$;
		\item[(ii)] strictly proper, if it is proper with $D = 0$. \hfill $\square$
	\end{itemize} 
\end{definition}
The above definitions of proper and strictly proper systems are inspired by~\citep{skogestad2005multivariable} and~\citep{desoer_generalized_1980}.
However, we note that sightly different definitions for properness and strict properness exist, cf.~\citep{macfarlane_poles_1976}.

\begin{definition}[Pole and zero polynomials~\citep{macfarlane_poles_1976}] \label{def:poles_zeros}~\\	
	\begin{itemize}
		\item[(i)] The pole polynomial $p(s)$ of a square transfer function $H(s)$ is defined as the least common denominator of all minors of all order of $H(s)$.
		\item[(ii)] The zero polynomial $z(s)$ of a square transfer function $H(s)$ is defined as the numerator of the rational function that is equivalent to $\det{H(s)}$ and that has $p(s)$ as its denominator, i.e.
		\item[] %
		\begin{equation} \label{eq:poles_zeros}
			\frac{z(s)}{p(s)} = \det{H(s)}.
		\end{equation}
		\hfill $\square$
	\end{itemize}
\end{definition}
The poles and zeros of a transfer function $H(s)$ are defined as the roots of the pole and zero polynomials, respectively. 
Observe that these definitions for the poles and zeros ensure that no pole-zero cancellations can occur between separate matrix components of $H(s)$.
That is, if $c \in \mathbb{C}$ is a root of the denominator of a component $H_{ij}(s)$, $i,j \in \{1,\dots,m\}$, then $c$ is a pole of $H(s)$.

\begin{definition}[Stable and minimum phase transfer functions]\label{def:stable_tf}
	A transfer function $H(s) \in \mathbb{C}^{m \times m}$ is said to be
	\begin{itemize}
		\item stable, if all its poles have negative real part and
		\item minimum phase, if all its zeros have negative real part.
	\end{itemize}
	\hfill $\square$
\end{definition}

\begin{definition}[Stable feedback interconnections]
	\label{def:stable_feedback_interconnection}
	Consider the feedback interconnection of two systems shown in Figure~\ref{fig:feedback_interconnection} with transfer functions $H_1(s) \in \mathbb{C}^{m \times m}$ and $H_2(s) \in \mathbb{C}^{m \times m}$, respectively.
	The feedback interconnection is said to be stable, if the closed-loop transfer function $H_3(s)$ with $y_2(s) \doteq H_3(s) r(s)$ is stable. \hfill $\square$
\end{definition}

\begin{figure}[t]
	\centering
	\begin{tikzpicture}
		\node[draw, minimum width=1.5cm, minimum height=1cm] (H1) at (0,0) {$H_1(s)$};
		\node[draw, circle, inner sep=0.8mm] (sum) at ($(H1.west)+(-0.75,0)$) {};
		\node[draw, minimum width=1.5cm, minimum height=1cm] (H2) at (3,0) {$H_2(s)$};
		\draw[-latex] (sum) -- node[above] {$u_1$} (H1);
		\draw[-latex] (H1) -- node[above] {$y_1 = u_2$} (H2);
		\draw[-latex] (H2.east) -| ($(H2.east)+(0.75,-2)$) -- ($(sum)+(0,-2)$) -- (sum) node[below right] {$-$};
		\draw[-latex] ($(sum)+(-0.75,0)$) node[above]{$r$} -- (sum);
		\draw[-latex] ($(H2.east)+(0.75,0)$) -- node[above] {$y_2$} ($(H2.east)+(1.5,0)$);
	\end{tikzpicture}
	\caption{Feedback interconnection of $H_1(s)$ and $H_2(s)$.}
	\label{fig:feedback_interconnection}
\end{figure}
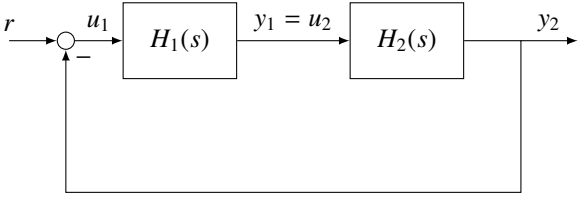

The main concern of this article are sufficient conditions on $H_1(s)$ and $H_2(s)$ that guarantee stability of the feedback interconnection shown in Figure~\ref{fig:feedback_interconnection}. 
Following the spirit of passivity theory, we seek conditions that can be verified independently for $H_1(s)$ and $H_2(s)$, respectively.

 \begin{remark}[Non-minimal state space realizations]
 	The careful reader may have noticed that we define stability via the poles of the transfer function $H(s)$ instead of the eigenvalues of the state matrix $A$ in~\eqref{eq:ss}.
 	This is motivated by power system applications where $H$ can serve as generalized impedance or admittance matrix-valued transfer functions~\citep{harnefors2015passivity}.
 	Crucially, we do not make any assumptions on the controllability of the pair $(A,B)$ or the observability of the pair $(A,C)$.
 	That is, we allow for non-minimal state space realizations~\eqref{eq:ss}.
 	Consequently, we only study the asymptotic behavior of $u_1(t)$, $u_2(t)$, $y_1(t)$, and $y_2(t)$ in Figure~\ref{fig:feedback_interconnection}, but we make no statements on the state evolutions $x_1(t)$ and $x_2(t)$, where $x_1$ and $x_2$ denote the states of systems one and two, respectively.
 	However, additional statements on the convergence of $x_1(t)$ and $x_2(t)$ can be obtained under additional zero-state detectability conditions~\cite[Def.~3.2.15]{vanderschaft_l-gain_2017}.
 	\hfill $\square$
\end{remark}

\section{Frequency-Dependent Passivity Analysis} \label{sec:main_result}

Inspired by applications in power electronics and power systems, we analyze the stability of the feedback interconnection shown in Figure~\ref{fig:feedback_interconnection} through so-called frequency-dependent passivity indices which we define next.

\subsection{Definition and Properties}

\begin{definition}[Frequency-dependent passivity indices] \label{def:freq-dependent-ofp-ifp}
	Consider a transfer function $H(s)$ and suppose that all zeros and poles of $H(s)$ have non-zero real part.
	Then, we define
	\begin{itemize}
		\item[(i)] the \gls{fifp} index of $H(s)$ at freqency $\omega \in \R_{\geq 0}$ as
		\begin{equation} \label{eq:def-ifp}
			\ifp(\omega) \doteq \frac{1}{2}\lmin{H(j\w) + H(j\omega)^\trans}.
		\end{equation}
		\item[(ii)] and the \gls{fofp} index of $H(s)$ at freqency $\omega \in \R_{\geq 0}$ as 
		\begin{equation}\label{eq:def-ofp}
			\ofp(\omega) \doteq \frac{1}{2} \lmin{H(j\w)^{-1} + (H(j\omega)^{-1})^{\trans}} 
		\end{equation}
		\hfill $\square$
	\end{itemize}
\end{definition}

Such frequency-dependent passivity indices are used in power electronics and power systems applications to assess the passivity properties of converters and grids in specific frequency ranges, cf.~\citep{harnefors2015passivity}.
However, as pointed out in Section~\ref{sec:intro}~and to the best of our knowledge, there does not yet exist a formal stability analysis considering these frequency-dependent passivity indices.
In comparison, there exists a wide array of stability results based on constant scalar, i.e. frequency-independent, passivity indices~\citep{sepulchre_constructive_1997,khalil_nonlinear_2002,lozano2013dissipative}.
We next summarize some properties of the frequency-dependent passivity indices that we use in the subsequent stability analysis.

\begin{remark}[Passivity indices and hermitian part of the transfer function]\label{rem:ifpofp}
	Note that $\ofp(\w) = \lmin{\sym{H(j\w)^{-1}} }$ and $\ifp(\w) = \lmin{\sym{H(j\w)} }$.
\end{remark}

\begin{proposition}[Sign of the \gls{fofp} and \gls{fifp} indices]\label{prop:same_sign_ifp_ofp}
	Consider a stable transfer function $H(s)$ whose zeros have non-zero real part and denote the associated \gls{fofp} and \gls{fifp} indices by $\ofp(\omega)$  and $\ifp(\omega)$, respectively.
	Then, for each frequency $\w \in \R_{\geq 0}$,
	\begin{align*}
		\ofp(\omega) > 0 &\iff \ifp(\omega) > 0 \quad \textnormal{and}\\
		\ofp(\omega) < 0 &\iff \ifp(\omega) < 0.
	\end{align*} \hfill $\square$
\end{proposition}

\begin{proof}
	We first prove in Part (a) of the proof that $\sym{H(j\w)}$ is positive (negative) definite if and only if $\sym{H(j\w)^{-1}}$ is positive (negative) definite.
	Then, we show in Part (b) that the statement of the proposition follows from~\autoref{def:freq-dependent-ofp-ifp}.
	
	Part (a): Consider a matrix $A \in \C^{n\times n}$ and suppose that $\sym{A}$ is positive definite.
	Then, \autoref{lem:pos_def_of_hermitian_implies_pos_real_part} asserts that all eigenvalues of $A$ have positive real part. Thus, $A$ is regular and  $\Re{x^\trans A x} > 0 $ for all $x \in \C^n \setminus \{0\}$. Further recall that $\Re{z} = (z+z^\trans)/2$ for all $z \in \C$ and hence
	\begin{equation} \label{eq:complex_number_real_part_identity}
		\Re{x^\trans A x} = \frac{1}{2}(x^\trans A x + (x^\trans A x)^\trans) =\Re{x^\trans A^\trans x} \quad \forall x \in \C^n.
	\end{equation}
	Let $y = Ax$ for any $x \in \C^n$ and note that $y \neq 0$ if and only if $x \neq 0$ because $A$ is regular.
	Then, inserting $y = Ax$ into \eqref{eq:complex_number_real_part_identity} yields
	\begin{equation*}
		\Re{y^\trans A^{-1} y } = \Re{x^\trans A^\trans A^{-1} A x} = \Re{x^\trans A^\trans x} = \Re{x^\trans A x} > 0
	\end{equation*} for all $x \in \C^n$.
	Hence,
	\begin{equation*}
		 y^\trans \sym{A^{-1}} y = \frac{1}{2} \left(y^\trans (A^{-1})^{\trans} y + y^\trans A^{-1} y\right) = \Re{y^\trans A^{-1} y } > 0 
	\end{equation*} 
	for all $y \in \C^n$, i.e., $\sym{A^{-1}}$ is positive definite.
	Note that the same arguments hold in the opposite direction, since $(A^{-1})^{-1} = A$. That is, if  $\sym{A^{-1}}$ is positive definite, so is $\sym{A}$.
	Additionally, the same arguments hold for negative definiteness due to \autoref{lem:pos_def_of_hermitian_implies_pos_real_part}.
	
	Part (b): Finally, we relate the last lines to the statement of the proposition via \autoref{def:freq-dependent-ofp-ifp}. We have that $\lmin{\sym{H(j\w)} }$ is positive (negative) if and only if $\lmin{\sym{H(j\w)^{-1}} }$ is positive (negative).
	Thus, by~\autoref{def:freq-dependent-ofp-ifp}, $\ofp(\omega)$ is positive (negative) if and only if $\ifp(\omega)$ is positive (negative), cf.~\autoref{rem:ifpofp}. 
\end{proof}

By Proposition~\ref{prop:same_sign_ifp_ofp}, a lack of \gls{ifp}, i.e. a negative sign of the \gls{fifp} index, implies a lack of \gls{ofp}, and vice versa.
The same holds analogously for an excess of passivity, i.e. positive signs of the \gls{fifp} and \gls{fofp} indices.

\begin{proposition}[\gls{fifp} and \gls{fofp} for negative frequencies] \label{prop:negative_frequencies}
	Consider a transfer function $H(s)$ and suppose that all zeros and poles of $H(s)$ have non-zero real part.
	Then, $\ofp(\w) = \ofp(-\w)$ and $\ifp(\w) = \ifp(-\w)$ for all $\w \in \R$.
\end{proposition}

\begin{proof}
	We first prove that $H(j\w) = (H(-j\w)^\trans)^\top$ for all $\w \in \R$. Then, we insert this into the definitions of $\ifp(\w)$ and $\ofp(\w)$ to prove the statements for \gls{fifp} and \gls{fofp}, respectively.
	
	Recall that $H(s)$ defined via~\eqref{eq:tf} is a rational transfer function matrix, i.e., all matrix elements of $H(s)$ are rational functions in $s$. Hence, $H(j\w) = (H(-j\w)^\trans)^\top$ for all $\w \in \R$.
	
	Inserting $H(j\w) = (H(-j\w)^\trans)^\top$ into the definition of $\ifp(\w)$ in~\eqref{eq:def-ifp} yields
	\begin{align}
		\nu(\w) &= \frac{1}{2} \lmin{(H(-j\w)^\trans)^\top + ((H(-j\w)^\trans)^\top)^\trans }. \label{eq:prop_pos_neg_freq}
 	\end{align} 
	Taking into account that the transposition and complex-conjugated transposition are commutative, we have for the second term in~\eqref{eq:prop_pos_neg_freq} that 
	\begin{equation} \label{eq:prop_pos_neg_freq_3}
		((H(-j\w)^\trans)^\top)^\trans = ((H(-j\w)^\trans)^\trans)^\top = H(-j\w)^\top.
	\end{equation}
	
	Inserting the right-hand side of \eqref{eq:prop_pos_neg_freq_3} into the second term of the $\lmin{}$ operator in~\eqref{eq:prop_pos_neg_freq}, we have
	\begin{align}
		\nu(\w)	&=\frac{1}{2} \lmin{(H(-j\w)^\trans)^\top + H(-j\w)^\top }. \label{eq:prop_pos_neg_freq_2}
 	\end{align}
	Observing that $\spec{A} = \spec{A^\top}$ for all $A \in \C^{n\times n}$, \eqref{eq:prop_pos_neg_freq_2} reads
	\begin{align*}
		\nu(\w) &=\frac{1}{2} \lmin{H(-j\w)^\trans + H(-j\w) } = \ifp(-\w),
 	\end{align*}
 	which completes the proof for the \gls{fifp} index. 
	
	Regarding the \gls{fofp} index: Note that since $H(s)$ is rational, $H(s)^{-1}$ is also rational, and hence $H(j\w)^{-1} = ((H(-j\w)^{-1})^\trans)^\top$. Inserting this into the definition of $\ofp(\w)$ in~\eqref{eq:def-ofp} yields
 	\begin{align*}
 		\ofp(\w) &= \frac{1}{2} \lmin{((H(-j\w)^{-1})^\trans)^\top + (((H(-j\w)^{-1})^{\trans})^\top)^\trans}\\
 		& = \ofp(-\w),
 	\end{align*}
	where the last step follows with the same arguments as above. This completes the proof.
\end{proof}

\subsection{Connection to Scalar Passivity Indices}

\begin{proposition}[Relation between scalar and frequency-dependent passivity indices]\label{prop:connection_frequency_ofp}
	Consider an LTI system~\eqref{eq:ss} and its transfer function $H(s)$.
	Suppose that all zeros and poles of $H(s)$ have non-zero real part.
	Then, the following holds:
	
	\begin{itemize}
		\item[(i)] If $H(s)$ is stable, then the largest index $\Ifp$ such that system~\eqref{eq:ss} is IFP is given by $\Ifp = \inf_{\omega \in \R} \ifp(\omega)$.
		\item[(ii)] If $H(s)^{-1}$ is stable, then the largest index $\Ofp$ such that system~\eqref{eq:ss} is OFP is given by $\Ofp = \inf_{\omega \in \R} \ofp(\omega)$.
	\end{itemize}
\end{proposition}

\begin{proof}
	The proof is given in \cite[Theorem 10 and 11]{xia_sector_2020}.
\end{proof}

\begin{remark}[Passivity indices for unstable and non-minimum phase systems]
	Note that as per \autoref{prop:stable_minimum_phase_inverse} in \ref{sec:appendix}, $H(s)$ has to be minimum phase in order for $H(s)^{-1}$ to be stable. Hence, the \gls{ofp} index $\Ofp$ is related to the \gls{fofp} index $\ofp(\w)$ if $H(s)$ is minimum phase, and the \gls{ifp} index $\Ifp$ is related to the \gls{fifp} index $\ifp(\w)$ if $H(s)$ is stable. However, the \gls{fifp} and \gls{fofp} indices are well-defined even if $H(s)$ has poles or zeros with positive real part, as long as all poles and zeros have non-zero real part. In this work, we use the \gls{fifp} and \gls{fofp} indices to ensure stability also in cases where scalar indices do not exist and where prior stability results based on scalar indices are thus not applicable. \hfill $\square$
\end{remark}

\subsection{Stability of Feedback Interconnections}

We proceed by deriving sufficient conditions on the frequency-dependent passivity indices of $H_1(s)$ and $H_2(s)$ which guarantee stability of the feedback interconnection shown in Figure~\ref{fig:feedback_interconnection}.
Our analysis relies on the well-known Nyquist criterion which we recall next.

\begin{theorem}[Nyquist criterion for open-loop stable systems] \label{th:generalized_nyquist}
	Consider the closed-loop system shown in Figure~\ref{fig:feedback_interconnection} and suppose that the rational open-loop transfer function $H(s) \doteq H_1(s) H_2(s) \in \C^{m\times m}$ is proper and stable.
	Then, the closed-loop system is exponentially stable if and only if 
	\begin{itemize}
		\item[(i)] for all $\omega \in [-\infty,\infty]$, the point $-1+j0$ is not an eigenvalue of $H(j\omega)$ and
		\item[(ii)] the eigenloci of $H(j\omega)$ as $\omega$ proceeds from $\omega=-j\infty$ to $\omega = j\infty$ make no counterclockwise encirclements of the point $-1+j0$. \hfill $\square$
	\end{itemize}
\end{theorem}

The above theorem is a specialized version of~\citep[Theorem L3]{desoer_generalized_1980} which we have adapted to open-loop stable systems and a gain $k=1$, where $k$ is in the notation of~\citep{desoer_generalized_1980}. 
Similar versions of the generalized Nyquist criterion with conditions on the determinant of $H(s)$ are discussed in \citep{griggs_interconnections_2012} or \citep{skogestad2005multivariable}.
The following theorem is the main result of this work. 

\begin{theorem}[Frequency-dependent passivity stability theorem] \label{th:main}
	Consider two proper systems with stable transfer functions $H_1(s)$ and $H_2(s)$, respectively, and suppose that all zeros of $H_1(s)$ and $H_2(s)$ have non-zero real part.
	Let the \gls{fifp} and \gls{fofp} passivity indices of $H_1(s)$ and $H_2(s)$ be $\ifp_1(\w)$, $\ofp_1(\w)$ and $\ifp_2(\w)$, $\ofp_2(\w)$, respectively. Let
	\begin{subequations} \label{eq:main_cond}
		\begin{align} 
			\ifp_1(\w) + \ofp_2(\w) &> 0 \quad \forall \w \in \R_{\geq 0} \quad \textnormal{and}\label{eq:main_cond_1}\\
			\ifp_2(\w) + \ofp_1(\w) &> 0 \quad \forall \w \in \R_{\geq 0}.\label{eq:main_cond_2}
		\end{align}
	\end{subequations}
	Then, the feedback interconnection of $H_1$ and $H_2$ is stable, if the conditions 
	\begin{equation}\label{item:cond-ii} 
		-1 \notin \spec{H_{\infty}}
	\end{equation}
	hold, where $H_{\infty} \doteq \lim_{\w \rightarrow \infty}H_1(j\w)H_2(j\w)$. 
\end{theorem}

\begin{remark}
	Note that condition \eqref{item:cond-ii} is automatically satisfied if at least one of the systems $H_1(s)$ or $H_2(s)$ is strictly proper, since then $H_{\infty} = H_1(j\w)H_2(j\w) = 0$.
\end{remark}

\begin{proof}
	The proof consists in showing that the conditons of the Nyquist criterion in \autoref{th:generalized_nyquist} hold.
	We proceed as follows: We first show in Part (a) of the proof that $L(j\w) \doteq H_1(j\w) H_2(j\w)$ has no real-valued eigenvalue $\tau \leq -1$ for all $\w \in \R$ due to condition~\eqref{eq:main_cond}. Subsequently, we show in Part (b) that this also holds for $\w \in \{-\infty,\infty\}$.
	Finally, we show in Part (c) that this implies that there is no crossing or encirclement of the -1 and invoke the Nyquist criterion from~\autoref{th:generalized_nyquist} to prove stability.
	
	Part (a): By~\autoref{prop:negative_frequencies}, \eqref{eq:main_cond} implies that
	\begin{subequations}
		\begin{align} 
			\ifp_1(\w) + \ofp_2(\w) &> 0 \quad \forall \w \in \R \quad \text{and}\label{eq:main_cond1_allR}\\
			\ifp_2(\w) + \ofp_1(\w) &> 0 \quad \forall \w \in \R.
		\end{align}
	\end{subequations}
	Note that
	\begin{equation*}
		\spec{H_1(j\w)H_2(j\w)} = \spec{H_2(j\w)H_1(j\w)},
	\end{equation*}	
	because the products $H_1(j\w)H_2(j\w)$ and $H_2(j\w)H_1(j\w)$ are similar matrices~\cite[Th.~1.3.22]{horn_matrix_2017}. 
	In the following, we will show that, for all $\w \in \R$, the matrix $L(j\w) = H_1(j\w) H_2(j\w)$ has no real-valued eigenvalue $\tau \leq -1$.
	Notice that, by Proposition~\ref{prop:same_sign_ifp_ofp}, conditions~\eqref{eq:main_cond} together imply that, for each $\w \in \R$, at least one of $\ofp_1(\w)$ and $\ofp_2(\w)$ must be positive. In the following, we thus distinguish two cases.
	
	Case $\ofp_2(\w)>0$: We first consider the case where $\ofp_2(\w) > 0$ and show that the matrix $H_2(j\w)H_1(j\w)$ has no real-valued eigenvalue ${\tau \leq -1}$.
	The proof proceeds by contradiction and so, for the time being, let us assume that $\tau$ is real valued and that $\tau\leq -1$.
	Consider the eigenvalue-eigenvector equation
	\begin{equation} \label{eq:eig_eq}
		H_2(j\w)H_1(j\w) v = \tau v,
	\end{equation}
	where $\tau \leq -1$ is a real-valued eigenvalue of $H_2(j\w)H_1(j\w)$ and where $v \in \C^{n}$ is the associated eigenvector.	
	We insert the definitions of the~\gls{fifp} and~\gls{fofp} indices into~\eqref{eq:main_cond1_allR} and obtain that, for all $\w \in \R$,
	\begin{equation} \label{eq:derivation_main_1}
		\lmin{H_1(j\w) +  H_1(j\w)^\trans} + \lmin{H_2(j\w)^{-1}  +  (H_2(j\w)^{-1})^{\trans}} > 0.
	\end{equation}
	Recall that, for any $A \in \C^{n\times n}$, $\sym{A} \doteq (A + A^\trans)/2$. Thus, we can rewrite \eqref{eq:derivation_main_1} as
	\begin{equation*}
		\lmin{\sym{H_1(j\w)}}  + \lmin{\sym{H_2(j\w)^{-1}}} > 0 \quad \forall \w \in \R.
	\end{equation*}

	With \autoref{prop:Rayleigh} from the Appendix, we thus obtain that
	\begin{equation*} 
		x^\trans \sym{H_1(j\w)} x + x^\trans \sym{H_2(j\w)^{-1}} x > 0 \quad \forall x \in \C^n.
	\end{equation*}
	Taking into account that $\sym{A} = A - \skew{A}$ for all $A \in \C^{n\times n}$, we obtain that, for all $x \in \C^n$ and for all $\w \in \R$,
	\begin{equation} \label{eq:def_cond_normal}
		x^\trans \left(H_1(j\w) \hspace*{0mm} - \hspace*{0mm} \skew{H_1(j\w)} \hspace*{0mm} + \hspace*{0mm} H_2(j\w)^{-1} \hspace*{0mm} - \hspace*{0mm} \skew{H_2(j\w)^{-1}} \right) x \hspace*{0mm} > \hspace*{0mm} 0.
	\end{equation}	
		
	Recall that, by the assumption stated in the theorem, all zeros of the transfer functions $H_1(s)$ and $H_2(s)$ have non-zero real part.
	Hence, for all $\w \in \R$, $H_1(j\w)$, $H_2(j\w)$, $H_1(j\w)H_2(j\w)$, and $H_2(j\w)H_1(j\w)$ are regular and thus $\tau \neq 0$. 
	Multiplying \eqref{eq:eig_eq} with $H_2(j\w)^{-1}$ from the left-hand side yields
	\begin{equation} \label{eq:eig_1}
		H_1(j\w) v = \tau H_2(j\w)^{-1} v.
	\end{equation}
	Choosing $x = v$ in the definiteness condition \eqref{eq:def_cond_normal} and inserting~\eqref{eq:eig_1} yields
	\begin{equation} \label{eq:contradiction_1}
		(1 + \tau)v^\trans H_2(j\w)^{-1} v - v^\trans \skew{H_1(j\w)} v - v^\trans \skew{H_2(j\w)^{-1}} v > 0.
	\end{equation}
	Since $\ifp_1(\w)+\ofp_2(\w) \in \R$, the left-hand side of \eqref{eq:contradiction_1} must be real valued.
	Recall that $\Re{x^\trans \skew{A} x} = 0$ for all $x \in \C^n$ and for all $A \in \C^{n\times n}$, cf.~\cite[Th.~4.1.3]{horn_matrix_2017}.
	Hence, $\Re{v^\trans \skew{H_1(j\w)}v + v^\trans \skew{H_2(j\w)^{-1}}v} = 0.$
	From~\eqref{eq:contradiction_1}, we thus obtain that
	\begin{equation} \label{eq:contradiction_2}
		(1 + \tau) \Re{v^\trans H_2(j\w)^{-1} v} > 0.
	\end{equation}
	However, if $\ofp_2(\w) = \lmin{\sym{H_2(j\w)^{-1}}}>0$, then $\Re{x^\trans H_2(j\w)^{-1} x}  = x^\trans \sym{H_2(j\w)^{-1}} x > 0$ for all $x \in \C^n$.
	Consequently, \eqref{eq:contradiction_2} cannot hold for any real-valued $\tau \leq -1$, i.e., the matrix $H_2(\w)H_1(\w)$ has no real valued eigenvalue $\tau \leq -1$ if $\ofp_2(\w) > 0$.
	
	Case $\ofp_1(\w)>0$: Repeating the same arguments for \eqref{eq:main_cond_2} leads to the conclusion that $H_1(j\w)H_2(j\w)$ has no real-valued eigenvalue $\tau \leq -1$ as long as $\ofp_1(\w) > 0$. Since for all $\w \in \R$, either $\ofp_1(\w)$ or $\ofp_2(\w)$ must be larger than zero, we conclude that $H_1(j\w)H_2(j\w)$ has no real-valued eigenvalue $\tau \leq -1$ for all $\w \in \R$. This completes part (a) of the proof.
		
	Part (b): With the assumption of proper transfer matrices, the limit matrix $\lim_{\w \to \infty} H_1(j\w)H_2(j\w)$ exist. If condition~\eqref{item:cond-ii} holds, then $H_1(j\w)H_2(j\w)$ has no real-valued eigenvalue $\tau = -1$ at $\w = \infty$. Furthermore, observe that $\lim_{\w \to \infty} H_1(j\w)H_2(j\w) = \lim_{\w \to -\infty} H_1(j\w)H_2(j\w)$, since the entries of $H_1(j\w)H_2(j\w)$ are proper rational functions which always have equal limits as $\w \to \pm\infty$. Hence, together with Part (a), $H_1(j\w)H_2(j\w)$ has no real-valued eigenvalue $\tau \leq -1$ for $\w \in [-\infty, \infty]$.
	
	Part (c): By combining the results from parts (a) and (b), we obtain that the eigenloci of $L(j\w)$ as $\w$ proceeds from $\w=-j\infty$ to $w=j\infty$ do not cross the point $-1+j0$.
	Furthermore, the eigenloci are continuous in $\w = [-\infty,\infty]$ and they do not cross the real axis to the left of the point $-1+j0$.
	Thus, the eigenloci do not encircle the point $-1+j0$.
	Moreover, since $H_1(s)$ and $H_2(s)$ are stable by the assumptions stated in the theorem, $L(s) = H_1(s)H_2(s)$ is also stable.
	Thus, we can apply the Nyquist Theorem~\ref{th:generalized_nyquist} and obtain that the feedback interconnection shown in Figure~\ref{fig:feedback_interconnection} is stable.
\end{proof}

The following corollary is a special case of \autoref{th:main}, in which one of the systems has an excess of \gls{ofp}, i.e. $\ofp(\w) > 0$ for all frequencies $\w \in \R_{\geq 0}$. Then, one of the two conditions in \eqref{eq:main_cond} can be dropped, and only one condition has to be checked to ensure stability of the feedback interconnection.

\begin{corollary}[Stable feedback interconnection with one passive system] \label{cor:main_result_passive}
	Consider two proper systems with stable transfer functions $H_1(s)$ and $H_2(s)$, respectively, and suppose that all zeros of $H_1(s)$ and $H_2(s)$ have non-zero real part. 
	Let at least one of the systems be strictly proper and denote the frequency-dependent \gls{ifp} and \gls{ofp} passivity indices of $H_1(s)$ and $H_2(s)$ as $\ifp_1(\w)$ and $\ofp_2(\w)$, respectively.
	Furthermore, suppose that $\ofp_2(\w) > 0$ for all $\w \in \R_{\geq 0}$.
	The feedback interconnection of $H_1(s)$ and $H_2(s)$ is stable if
	\begin{equation} \label{cond_corollary}
		\nu_1(\w) + \rho_2(\w) > 0
	\end{equation}
	holds for all frequencies $\w \in \R_{\geq 0}$.
\end{corollary}

\begin{proof}
	By the assumptions stated in the theorem, we have that $\rho_2(\w) > 0$ for all $\w \geq 0$. This implies, following \autoref{lem:pos_def_of_hermitian_implies_pos_real_part}, that all eigenvalues of $H_2(j\w)^{-1}$ have strictly positive real parts for all $\w \in \R$. Thus, condition \eqref{cond_corollary} alone guarantees that $H_1(j\w)H_2(j\w)$ has no real-valued eigenvalues $\tau \leq -1$. Stability of the feedback interconnection then follows via the same arguments as in \autoref{th:main}.
\end{proof}

\begin{remark}
	The condition in \autoref{cor:main_result_passive} is equivalent to the conditions presented in \cite{chen_limitations_2024} (and references therein). Consequently, the existing literature conditions represent a special case of \autoref{th:main}, and our results provide greater generality.
\end{remark}

\subsection{Interpretation and Comparison to Existing Passivity-based Feedback Stability Theorems}
In this section, we compare \autoref{th:main} to a prior, well known passivity-based feedback stability theorem which we recall next.

\begin{proposition}[{Stability under scalar passivity indices~\cite[Theorem 2.33]{sepulchre_constructive_1997}}] \label{prop:existing_stability_theorems}
	Consider two LTI systems of form~\eqref{eq:ss} with transfer functions $H_1(s)$ and $H_2(s)$, respectively. 
	Suppose that the systems are \gls{if-ofp} with indices $\Ofp_1, \Ifp_1 \in \R$ and $\Ofp_2, \Ifp_2 \in \R$, respectively.
	If  
	\begin{subequations} \label{eq:existing_conditions}
		\begin{align}
			\Ofp_1 + \Ifp_2 &> 0 \\
			\Ofp_2 + \Ifp_1 &> 0,
		\end{align}
	\end{subequations}
	then the feedback interconnection shown in Figure~\ref{fig:feedback_interconnection} is stable.
\end{proposition}

\begin{remark}[Feedback stability theorems based on passivity] \label{rem:existing_stability_theorems}
	Proposition~\ref{prop:existing_stability_theorems} can be found in standard textbooks, e.g. in~\citep[Theorem 2.33]{sepulchre_constructive_1997} where it is stated for nonlinear state space systems.
	However, there exist equivalent variants for general input-output maps,~cf.~\cite[Theorem 2.2.18]{vanderschaft_l-gain_2017}.
	An overview of results based on different passivity notions is presented in~\cite[Chapter 5]{lozano2013dissipative}. \hfill $\square$
\end{remark}

Note that Theorem~\ref{th:main} is restricted to LTI systems.
In contrast, the prior results listed in Remark~\ref{rem:existing_stability_theorems} also apply to nonlinear systems.
However, for LTI systems, Theorem~\ref{th:main} allows for frequency-dependent passivity compensation, whereas theorems based on scalar passivity indices require a compensation of the worst-case lack and excess of passivity of both systems to infer stability.
This can be seen by comparing Conditions~\eqref{eq:main_cond} of Theorem~\ref{th:main} and Conditions~\eqref{eq:existing_conditions} of \autoref{prop:existing_stability_theorems}.
By \autoref{prop:connection_frequency_ofp}, the scalar \gls{ifp} and \gls{ofp} passivity indices in \eqref{eq:existing_conditions} are bounded above by the infimum of the frequency-dependent passivity indices \gls{fifp} and \gls{fofp} in \eqref{eq:main_cond}, respectively.
This is a key aspect of the novelty of our result, which allows to guarantee stability in cases where conditions~\eqref{eq:existing_conditions} on the scalar passivity indices do not hold.
Consequently, \autoref{th:main} provides a less conservative approach when studying feedback interconnections of LTI systems than prior results based on scalar passivity indices.
The price to pay, however, is that \autoref{th:main} does not address feedback interconnections of nonlinear systems.

\section{Illustrative Example} \label{sec:example}

Consider the feedback interconnection shown in Figure~\ref{fig:feedback_interconnection} with transfer functions

\begin{align*}
    H_1(s) &= \begin{bmatrix} \frac{s + 2}{s + 1} & \frac{0.55}{s^2 + 0.2s + 0.5} \\ 0 & \frac{s + 2}{s + 1} \end{bmatrix} \quad \text{and} \\
    H_2(s) &= \begin{bmatrix} \frac{s - 0.2}{s^2 + 1.2s + 0.6} & \frac{-2.7s-4}{s + 4} \\ 1 & \frac{s + 0.2}{s^2 + 1.3s + 0.1} \end{bmatrix}
\end{align*}
which are stable, do not have purely imaginary zeros, and meet condition~\eqref{item:cond-ii} of \autoref{th:main}. 
To infer stability of the feedback interconnection via \autoref{th:main}, we next assess whether conditions~\eqref{eq:main_cond} of \autoref{th:main} hold, i.e., that the sum of the frequency-dependent passivity indices of $H_1(s)$ and $H_2(s)$ is positive for all frequencies $\w \in \R_{\geq 0}$.

The frequency-dependent passivity indices of both systems over the frequency range $\omega \in [10^{-3}, 10^{4}]$ rad/s are depicted in Figures~\ref{fig:indices_H1} and~\ref{fig:indices_H2}, respectively.
It can be seen that both systems exhibit frequency ranges with negative frequency-dependent passivity indices. Thus, according to \autoref{prop:connection_frequency_ofp}, both systems have a lack of \gls{ifp} and \gls{ofp} in terms of their scalar passivity indices $\Ofp$ and $\Ifp$.
However, Figure~\ref{fig:indices_sum} shows that the sum of the frequency-dependent passivity indices of both systems stays positive over the whole frequency range.
Consequently, the conditions of \autoref{th:main} are satisfied and the feedback interconnection is guaranteed to be stable.

\begin{figure}[h]
    \centering
    \includegraphics{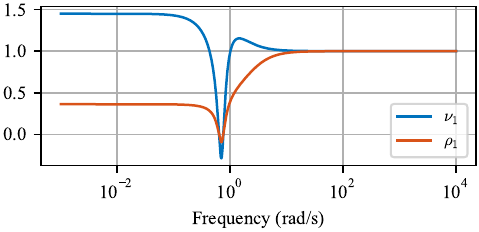}
    \caption{Frequency-dependent passivity indices of system $H_1(s)$.}
    \label{fig:indices_H1}
\end{figure}

\begin{figure}[h]
    \centering
    \includegraphics{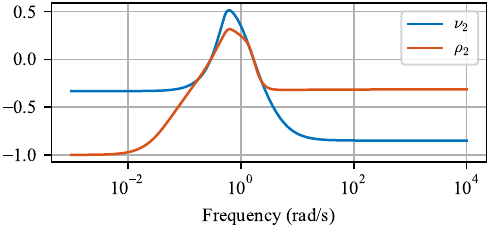}
    \caption{Frequency-dependent passivity indices of system $H_2(s)$.}
    \label{fig:indices_H2}
\end{figure}

\begin{figure}[h]
    \centering
    \includegraphics{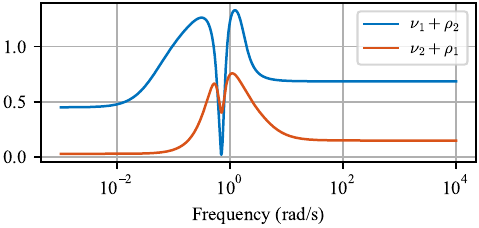}
    \caption{Sum of frequency-dependent passivity indices of systems $H_1(s)$ and $H_2(s)$.}
    \label{fig:indices_sum}
\end{figure}

In this small-scale example, $H_1(s)$ and $H_2(s)$ are non-passive, but the frequency ranges where their frequency-dependent passivity indices are negative do not overlap.
The frequency-dependent passivity indices of system $H_1(s)$ are positive everywhere except for frequencies close to 0.6 rad/s, and vice versa for $H_2(s)$.
Even though both systems are non-passive, their feedback interconnection is stable, as guaranteed by~\autoref{th:main}.
As we show next, the given example is not covered by existing feedback stability theorems based on passivity theory such as \autoref{prop:existing_stability_theorems} and \autoref{rem:existing_stability_theorems} which rely on scalar passivity indices.

In contrast to the \gls{ofp} and \gls{ifp} cases, the scalar passivity indices $\Ofp$ and $\Ifp$ of an \gls{if-ofp} system are not unique, because the inequality in Definition~\ref{def:passivity}~(iv) may hold for various combinations of $\Ofp$ and $\Ifp$.
Figure~\ref{fig:scalar_indices_1} shows such combinations of $\Ofp_1$ and $\Ifp_1$ for which $H_1(s)$ is \gls{if-ofp}.
That is, given a sequence of $40$ values for $\Ifp_1$, the plot shows the respective maximum value for $\Ofp_1$ such that $H_1(s)$ is \gls{if-ofp}, which we compute by solving a semidefinite program in the time domain resulting from \cite[Proposition 4.1.2]{vanderschaft_l-gain_2017}.
Likewise, Figure~\ref{fig:scalar_indices_2} shows the maximum value for $\Ifp_2$ as function of $\Ofp_2$ such that $H_2(s)$ is \gls{if-ofp}.
For both systems, the figures show that increasing the excess of \gls{ifp} comes at a cost of increasing the lack of \gls{ofp}.

To assess whether any of the sampled $(\Ifp_1,\Ofp_1)$ and $(\Ifp_2,\Ofp_2)$ pairs from Figures~\ref{fig:scalar_indices_1} and~\ref{fig:scalar_indices_2} meet the conditions of \autoref{prop:existing_stability_theorems}, we define the indicator function
\[
\mathcal{I} \doteq
\begin{cases}
    2, & \text{if } \Ofp_1 + \Ifp_2 > 0 \; \text{ and } \Ofp_2 + \Ifp_1 > 0, \\
	1, & \text{if either } \Ofp_1 + \Ifp_2 > 0 \; \text{ or } \; \Ofp_2 + \Ifp_1 > 0, \\
	0, & \text{else.}
\end{cases}
\]
Figure~\ref{fig:indicator} shows that either $\mathcal{I} = 0$ or $\mathcal{I} = 1$ for the 1600 combinations resulting of the pairs in Figures~\ref{fig:scalar_indices_1} and~\ref{fig:scalar_indices_2}, i.e., it appears that the conditions of \autoref{prop:existing_stability_theorems} do not hold for the considered example.

This demonstrates that the stability of the feedback interconnection cannot be explained by existing passivity-based stability theorems with constant indices for all frequencies, but it can be explained by \autoref{th:main} with frequency-dependent indices.

\begin{figure}[h]
    \centering
    \includegraphics{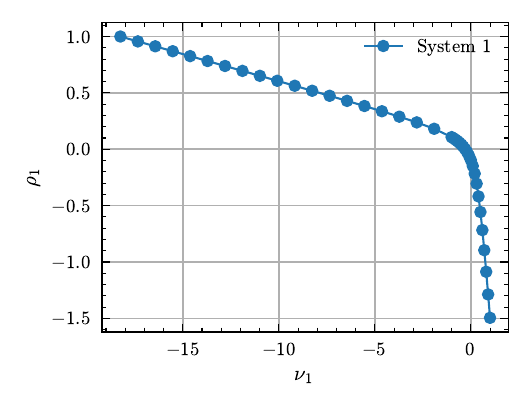}
    \caption{Scalar passivity indices for $H_1(s)$}
    \label{fig:scalar_indices_1}
\end{figure}

\begin{figure}[h]
    \centering
    \includegraphics{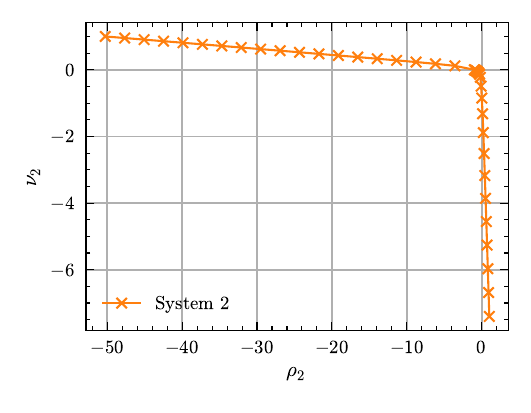}
    \caption{Scalar passivity indices for $H_2(s)$}
    \label{fig:scalar_indices_2}
\end{figure}

\begin{figure}[h]
    \centering
    \includegraphics{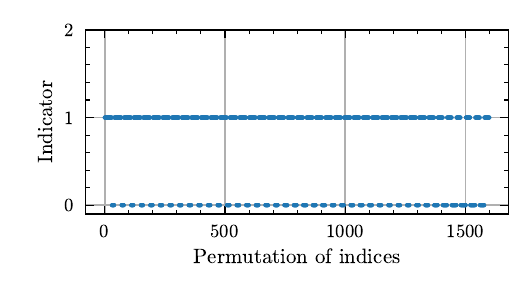}
    \caption{Indicator}
    \label{fig:indicator}
\end{figure}
\section{Conclusions and Future Work}

In this paper, we have developed a novel feedback stability theorem for linear time-invariant systems based on frequency-dependent passivity indices. We have demonstrated, that the proposed theorem is less conservative than existing stability theorems based on scalar passivity indices. In particular, the proposed theorem allows both systems to be non-passive, by requiring only a frequency-dependent compensation of lack and excess of passivity. This extends classic passivity-based stability theorems, which require a compensation on the scalar passivity indices for all frequencies.

The proposed stability conditions do not require full system-level analysis and can be verified by each subsystem individually using only the own subsystem dynamics. This allows the application of the theorem to areas where system dynamics of the subsystems are partially unknown, e.g. in power systems. Future work will focus on extending the proposed stability theorem to arbitrary interconnections of systems, and on the design of compensators for achieving desired frequency-dependent passivity indices.

{
\bibliographystyle{elsarticle-harv} 
\bibliography{references.bib}
}

\appendix

\section{Auxiliary Propositions} \label{sec:appendix}

\begin{proposition}[Rayleigh quotient bounds] \label{prop:Rayleigh}
	Let $A\in \C^{n \times n}$ be a hermitian matrix. Then,
	\begin{equation}
		\lambda_{\min}(A) \leq \frac{ x^\trans A x }{x^\trans x}\leq \lambda_{\max}(A)  \quad \forall x \in \C^n.
	\end{equation}
	\hfill $\square$
\end{proposition}

\begin{proof}
	A proof is given in \cite[Th.~4.2.2]{horn_matrix_2017}. Equality holds if $x$ is an eigenvector corresponding to the smallest (or largest) eigenvalue $\lambda_{\min}(A)$ (or $\lambda_{\max}(A)$).
\end{proof}

\begin{proposition}[Eigenvalue bounds] \label{prop:real_part_eig_bounded}
	Consider a matrix $A \in \C^{n\times n}$ and its hermitian part $\sym{A}$. The real part of the eigenvalues of $A$ are bounded by the smallest and largest eigenvalue of $\sym{A}$, i.e.
	\begin{equation} \label{eq:statement_real_part_eig_bounded}
		\Re{\lambda(A)} \in [\lambda_{\min}(\sym{A}), \lambda_{\max}(\sym{A})]
	\end{equation}
	for all eigenvalues $\lambda(A)$ of $A$. \hfill $\square$
\end{proposition}

\begin{proof}
	Let $\lambda \in \mathbb{C}$ be an eigenvalue of $A$ with corresponding eigenvector $v \in \mathbb{C}^n$, i.e. $Av = \lambda v$.
	Multiplying from the left-hand side with $v^\trans$ yields $v^\trans A v = \lambda v^\trans v$.
	Dividing both sides by $v^\trans v$ and taking the real part gives leads to
	\begin{equation} \label{eq:derivation_1}
		\Re{\lambda} = \frac{\Re{v^\trans A v}}{v^\trans v}.
	\end{equation}
	Observe that 
	\[\Re{v^\trans A v} = \frac{1}{2}\biggl(v^\trans A v + (v^\trans A v)^\trans \biggr) = v^\trans \sym{A} v.\]
	Inserting the last equation into~\eqref{eq:derivation_1} yields
	\begin{equation}\label{eq:derivation_2}
		\Re{\lambda} = \frac{v^\trans \sym{A} v}{v^\trans v}.
	\end{equation}
	Observe that the right-hand side of~\eqref{eq:derivation_2} is a Rayleigh quotient of the hermitian matrix $\sym{A}$. Thus, by applying Proposition~\ref{prop:Rayleigh}, we obtain statement~\eqref{eq:statement_real_part_eig_bounded}.
\end{proof}

\begin{corollary}[Hermitian part and positive definiteness] \label{lem:pos_def_of_hermitian_implies_pos_real_part}
	Consider a matrix $A \in \C^{n\times n}$ and its hermitian part $\sym{A}$. If $\sym{A}$ is positive (semi)definite, then all eigenvalues of $A$ have a positive (nonnegative) real part. \hfill $\square$
\end{corollary}

\begin{proof}
	If $\sym{A}$ is positive (semi)definite, then $\lambda_{\min}(\sym{A}) > 0$ ($\lambda_{\min}(\sym{A}) \geq 0$), and thus by \eqref{eq:statement_real_part_eig_bounded} all eigenvalues of $A$ have a positive (nonnegative) real part.
	Thus, the statement is a direct consequence of \autoref{prop:real_part_eig_bounded}.
\end{proof}

\begin{remark}[Hermitian part and negative definiteness]
	Note that \autoref{lem:pos_def_of_hermitian_implies_pos_real_part} also holds in the negative (semi)definite case, i.e., if $\sym{A}$ is negative (semi)definite, then all eigenvalues of $A$ have a negative (nonpositive) real part. \hfill $\square$
\end{remark}

\begin{proposition}[Stability of the inverse transfer function] \label{prop:stable_minimum_phase_inverse}
	Consider a transfer function matrix $H(s)\in\mathbb{C}^{n\times n}$.
	If $H(s)$ is minimum phase, then $H(s)^{-1}$ is stable.
\end{proposition}

\begin{proof}
	The inverse of a matrix $H(s)$ can be computed via the adjugate matrix $\adj{H(s)}$ and the determinant $\det{H(s)}$ as \cite[Sec.~0.8.3]{horn_matrix_2017}
	\begin{equation} \label{eq:inverse_tf}
		H(s)^{-1} = \frac{1}{\det{H(s)}} \adj{H(s)} \stackrel{\eqref{eq:poles_zeros}}{=} \frac{p(s)}{z(s)} \adj{H(s)},
	\end{equation}
	where $p(s)$ and $z(s)$ are the pole and zero polynomials of $H(s)$, respectively, as in \autoref{def:poles_zeros}. The adjugate matrix $\adj{H(s)}$ is composed of all the minors of $H(s)$. As per \autoref{def:poles_zeros}, the pole polynomial $p(s)$ is defined as the least common denominator of all minors of $H(s)$. Thus, no entry of the matrix $p(s) \adj{H(s)}$ has any poles. As a consequence, all the poles of the inverse transfer function \eqref{eq:inverse_tf} are given by $z(s)$. If the $H(s)$ is minimum phase, then $z(s)$ has no roots in the closed right half-plane, and thus $H(s)^{-1}$ in \eqref{eq:inverse_tf} is stable.
\end{proof}

\end{document}

\endinput